\documentclass[conference,a4paper]{IEEEtran}

\usepackage{amsfonts}
\usepackage{cite} 
\usepackage{amsmath,amsthm} 
\usepackage{amssymb}  

\newtheorem{theorem}{Theorem}%[section]
\newtheorem{lemma}{Lemma}
\newtheorem{definition}{Definition}

\newtheorem{remark}{Remark}

\newtheorem{assumption}{Assumption}

\newcommand{\sr}{\stackrel}

\newcommand{\tri}{\sr{\triangle}{=}}

\newcommand{\noi}{\noindent}

\newcommand{\be}{\begin{equation}}
\newcommand{\ee}{\end{equation}}
\newcommand{\bea}{\begin{eqnarray}}
\newcommand{\eea}{\end{eqnarray}}
\newcommand{\bes}{\begin{eqnarray*}}
\newcommand{\ees}{\end{eqnarray*}}

\newcommand{\bfi}{\begin{figure}}
\newcommand{\bfit}{\begin{figure}[t]}
\newcommand{\bfib}{\begin{figure}[b]}
\newcommand{\bfih}{\begin{figure}[h]}
\newcommand{\bfip}{\begin{figure}[p]}
\newcommand{\efi}{\end{figure}}

\newcommand{\bi}{\begin{itemize}}
\newcommand{\ei}{\end{itemize}}
\newcommand{\ben}{\begin{enumerate}}
\newcommand{\een}{\end{enumerate}}

\renewcommand{\IEEEQED}{\IEEEQEDclosed}

\begin{document}

\sloppy

%% Paper Title
%% You can use linebreaks \\ within to get better formatting as
%% desired. 
\title{Variational Equalities of Directed Information and Applications}

\author{
  \IEEEauthorblockN{Photios A. Stavrou and Charalambos D. Charalambous}
  \IEEEauthorblockA{Dep. of Electr. \& Comp. Eng., UCY, Nicosia, Cyprus\\
    {\it Email: \{stavrou.fotios,chadcha\}@ucy.ac.cy} }}

%% Create the title:
\maketitle

\begin{abstract}
In this paper we introduce two variational equalities of directed information, which are analogous to those of mutual information employed in the Blahut-Arimoto Algorithm (BAA). Subsequently, we introduce nonanticipative Rate Distortion Function (RDF) ${R}^{na}_{0,n}(D)$ defined via directed information introduced in \cite{charalambous-stavrou-ahmed2012a}, and we establish its equivalence to Gorbunov-Pinsker's nonanticipatory $\epsilon$-entropy $R^{\varepsilon}_{0,n}(D)$. By invoking certain results we first establish existence of the infimizing reproduction distribution for ${R}^{na}_{0,n}(D)$, and then we give its implicit form for the stationary case. Finally, we utilize one of the variational equalities and the closed form expression of the optimal reproduction distribution to provide an algorithm for the computation of ${R}^{na}_{0,n}(D)$.   
\end{abstract}

\section{Introduction}
\par Directed information from a sequence of Random Variables (RV's) $X^n\tri\{X_0,X_1,\ldots,X_n\}\in{\cal X}_{0,n}\tri\times_{i=0}^n{\cal X}_i$, to another synchronized sequence $Y^n\tri\{Y_0,Y_1,\ldots,Y_n\}\in{\cal Y}_{0,n}\tri\times_{i=0}^n{\cal Y}_i$ is, in general,  a  functional of two collections of nonanticipative or causal conditional distributions $\{P_{X_i|X^{i-1},Y^{i-1}}(\cdot|\cdot,\cdot),~P_{Y_i|Y^{i-1},X^i}(\cdot|\cdot,\cdot)~:~i=0,1,\ldots,n\}$, unlike mutual information which is a function of $P_{X^n}$ and $P_{Y^n|X^n}$.
In the past, directed information or its variants were used to characterize capacity of channels with memory and feedback  \cite{chen-berger2005,tatikonda-mitter2009}, lossy data compression of sequential codes \cite{tatikonda2000}, lossy data compression of block codes \cite{ma-ishwar2011}, and capacity of networks \cite{kramer1998}. 

In this paper, we adopt the mathematical formulation introduced in \cite{charalambous-stavrou2012}, to define directed information via relative entropy with respect to two consistent families of conditional distributions defined on abstract spaces, and we derive two variational equalities,  which are analogous to those of mutual information utilized in Blahut-Arimoto algorithm (BAA).

Subsequently, we introduce nonanticipative RDF, ${R}^{na}_{0,n}(D)$, introduced in \cite{charalambous-stavrou-ahmed2012a} to derive realizable filters,  and we show its relation to Gorbunov-Pinsker's nonanticipatory $\epsilon$-entropy, $R^{\varepsilon}_{0,n}(D)$. We then proceed by giving general conditions for existence of a minimizing nonanticipative reproduction distribution for ${R}^{na}_{0,n}(D)$, and we derive its implicit  form for the stationary case. Finally, we invoke one of the variational equalities and the implicit form of the optimal reproduction distribution  to present an algorithm for computing ${R}^{na}_{0,n}(D)$ similar to the BAA.  

Recently, BAA's are presented in \cite{naiss-permuter2013b} for lossy compression with feedforward at the decoder, and in \cite{naiss-permuter2013a} for feedback channel capacity without using the variational equalities derived in this paper. The fundamental difference between $R^{na}_{0,n}(D)$ and feedforward information RDF is that the former is nonanticipative while the latter need not to be nonanticipative.

Our interest in nonanticipative RDF, ${R}^{na}_{0,n}({D})$, is motivated by applications in which the processing of information  is done via symbol-by-symbol transmission (zero delay). Some applications are listed below.  \\
{\bf(1)} Source-channel matching via symbol-by-symbol transmission. A necessary condition for such matching is realizability of  the optimal reproduction distribution of RDF via an encoder-channel-decoder which are nonanticipative maps (operate causally) \cite{gastpar2003,kourtellaris-charalambous-stavrou2013a,charalambous-stavrou-ahmed2012a}. Therefore, for sources with memory the nonanticipative RDF is the appropriate information measure of lossy compression for source-channel matching via symbol-by-symbol transmission.\\
{\bf(2)} Computation of  the optimal performance theoretically attainable (OPTA) by sequential quantizers \cite{tatikonda2000}, see also \cite{ma-ishwar2011} for video coding applications.\\
{\bf(3)} Computation of upper bounds on the OPTA by non-causal codes. This follows from the equivalence between $R^{na}_{0,n}(D)$ and $R^{\varepsilon}_{0,n}(D)$ established in this paper.\\
{\bf (4)} Constructing realizable filters based on nonanticipative RDF (see \cite{charalambous-stavrou-ahmed2012a,stavrou-charalambous2013}).

The paper is structured as follows. In Section~\ref{causal_channels} we construct the two equivalent definitions of nonanticipative channels on abstract spaces, and we define directed information via the information divergence. In  Section~\ref{variational} we  derive the variational equalities of directed information. Finally, in  Section~\ref{applications} we give the connection between $R^{na}_{0,n}(D)$ and $R^{\varepsilon}_{0,n}(D)$, we give the implicit form of the optimal nonanticipative reproduction distribution for the stationary case, and we discuss an application of the variational equality to nonanticipative RDF. Lengthy proofs are omitted and references are given where they can be found.
\section{Nonancticipative Channels and Directed Information}\label{causal_channels}
\par In this section we define directed information using relative entropy, as a functional of two consistent families of conditional distributions that uniquely define $\{P_{X_i|X^{i-1},Y^{i-1}}(\cdot|\cdot,\cdot):i=0,1,\ldots,n\}$ and $\{P_{Y_i|Y^{i-1},X^i}(\cdot|\cdot,\cdot):i=0,1,\ldots,n\}$, respectively, and vice versa following \cite{charalambous-stavrou2012}. Throughout the paper we assume ${\cal X}_n$, ${\cal Y}_n$, $n=0,1,\ldots$, are Polish spaces.\\
{\bf Notation}. Let $\mathbb{N} \tri \{0,1,2,\ldots\},$ and $\mathbb{N}^n \tri \{0,1,2,\ldots,n\}.$ Introduce two sequence of measurable spaces $\{({\cal X}_n,{\cal B}({\cal X }_n)):n\in\mathbb{N}\}$ and $\{({\cal Y}_n,{\cal B}({\cal Y}_n)):n\in\mathbb{N}\},$ where ${\cal B}({\cal X}_n)$ and ${\cal B}({\cal Y}_n)$ are Borel $\sigma-$algebras of subsets of ${\cal X}_n$ and ${\cal Y}_n$, respectively. Points in ${\cal X}^{\mathbb{N}}\tri{{\times}_{n\in\mathbb{N}}}{\cal X}_n,$ ${\cal Y}^{\mathbb{N}}\tri{\times_{n\in\mathbb{N}}}{\cal Y}_n$ are denoted by ${\bf x}\tri\{x_0,x_1,\ldots\}\in{\cal X}^{\mathbb{N}},$ ${\bf y}\tri\{y_0,y_1,\ldots\}\in{\cal Y}^{\mathbb{N}}$, and their restrictions to finite coordinates by $x^n\tri\{x_0,x_1,\ldots,x_n\}\in{\cal X}_{0,n},$ $y^n\tri\{y_0,y_1,\ldots,y_n\}\in{\cal Y}_{0,n},$ for $n\in\mathbb{N}$.
Let ${\cal B}({\cal X}^{\mathbb{N}})\tri\odot_{i\in\mathbb{N}}{\cal B}({\cal X}_i)$, ${\cal B}({\cal Y}^{\mathbb{N}})\tri\odot_{i\in\mathbb{N}}{\cal B}({\cal Y}_i)$ denote the $\sigma-$algebras on ${\cal X}^{\mathbb{N}}$, ${\cal Y}^{\mathbb{N}}$, respectively, generated by cylinder sets. Hence, ${\cal B}({\cal X}_{0,n})$ and ${\cal B}({\cal Y}_{0,n})$ denote the $\sigma-$algebras of cylinder sets in ${\cal X}^{\mathbb{N}}$ and ${\cal Y}^{\mathbb{N}},$ respectively, with bases over $A_i\in{\cal B}({\cal X}_i)$,  $B_i\in{\cal B}({\cal Y}_i), i=0,1,\ldots,n$, respectively.
The set of stochastic kernels on ${\cal Y}$ given ${\cal X}$ is denoted by ${\cal Q}({\cal Y};{\cal X})$.\\
\noi{\bf Feedback Channel.} Suppose for each $n\in\mathbb{N},$ the distributions $\{p_n(dx_n;x^{n-1},y^{n-1}):n\in\mathbb{N}\}$ with $p_0(dx_0;x^{-1},y^{-1})\tri{p}_0(x_0)$ satisfy the following conditions.\\
{\bf i)} For $n\in\mathbb{N},$ $p_n(\cdot;x^{n-1},y^{n-1})$ is a probability measure on ${\cal B}({\cal X }_n);$\\
{\bf ii)} For every $A_n\in{\cal B}({\cal X}_n),~n\in\mathbb{N},$ $p_n(A_n;x^{n-1},y^{n-1})$ is a $\odot^{n-1}_{i=0}\big({\cal B}({\cal X }_i)\odot{\cal B}({\cal Y}_i)\big)$-measurable function of $x^{n-1}\in{\cal X}_{0,n-1},$ $y^{n-1}\in{\cal Y}_{0,n-1}$.\\
Let $C\in{\cal B}({\cal X}_{0,n})$ be a cylinder set of the form $
C\tri\big\{{\bf x}\in{\cal X}^{\mathbb{N}}:x_0\in{C_0},x_1\in{C_1},\ldots,x_n\in{C_n}\big\},~C_i\in{\cal B}({\cal X }_i),~i\in\mathbb{N}^n$.
Define a family of measures ${\bf P}(\cdot|{\bf y})$ on ${\cal B}({\cal X}^{\mathbb{N}})$ by
\begin{align}
{\bf P}(C|{\bf y})&\tri\int_{C_0}p_0(dx_0)\ldots\int_{C_n}p_n(dx_n;x^{n-1},y^{n-1})\label{equation2}\\
&\equiv{\overleftarrow{P}}_{0,n}(C_{0,n}|y^{n-1}),~C_{0,n}=\times_{i=0}^n{C_i}.\label{equation4a}
\end{align}
The notation ${\overleftarrow{P}}_{0,n}(\cdot|y^{n-1})$ denotes the restriction of the measure ${\bf P}(\cdot|{\bf y})$ on cylinder sets $C\in{\cal B}({\cal X}_{0,n})$, for $n\in\mathbb{N}$.\\
Thus, if conditions {\bf i)} and {\bf ii)} hold then for each ${\bf y}\in{\cal Y}^{\mathbb{N}},$ the right hand side of (\ref{equation2}) defines a consistent family of finite-dimensional distribution on $({\cal X}^{\mathbb{N}},{\cal B}({\cal X}^{\mathbb{N}}))$, and hence there exists a unique measure on $({\cal X}^{\mathbb{N}},{\cal B}({\cal X}^{\mathbb{N}})),$ from which $p_n(dx_n;x^{n-1},y^{n-1})$ is obtained. This is the usual definition of a feedback channel (its input distribution), as a family of functions $p_n(dx_n;x^{n-1},y^{n-1})$ satisfying conditions {\bf i)} and {\bf ii)}.\\
An alternative, equivalent definition of a feedback channel is established as follows. Consider a family of measures ${\bf P}(\cdot|{\bf y})$ on $({\cal X}^{\mathbb{N}},{\cal B}({\cal X}^{\mathbb{N}}))$ satisfying the following consistency condition.\\
{\bf C1}:~~If $E\in{\cal B}({\cal X}_{0,n})$, then ${\bf P}(E|{\bf y})$ is ${\cal B}({\cal Y}_{0,n-1})-$measurable function of ${\bf y}\in{\cal Y}^{\mathbb{N}}$.\\
The set of such measures is denoted by ${\cal Q}^{\bf C1}({\cal X}^{\mathbb{N}};{\cal Y}^{\mathbb{N}})$. For Polish spaces, it can be shown that for any family of measures ${\bf P}(\cdot|{\bf y})$ satisfying {\bf C1} one can construct a collection of conditional distributions  $\{p_n(dx_n;x^{n-1},y^{n-1}):n\in\mathbb{N}\}$ satisfying conditions {\bf i)} and {\bf ii)} which are connected with ${\bf P}(\cdot|{\bf y})$ via relation (\ref{equation2}).\\
\noi{\bf Feedforward Channel.} The previous methodology can be repeated for the collection of distributions $\{q_n(dy_n;y^{n-1},x^n):n\in\mathbb{N}\}$ which satisfy similar conditions to {\bf i)} and {\bf ii)}.
Similarly as before, define a family of measures ${\bf Q}(\cdot|{\bf x})$ on ${\cal B}({\cal Y}^{\mathbb{N}})$ by
\begin{align}
{\bf Q}(D|{\bf x})
&\tri\int_{D_0}q_0(dy_0;x_0)\ldots\int_{D_n}q_n(dy_n;y^{n-1},x^n)\label{equation4}\\
&\equiv{\overrightarrow{Q}}_{0,n}(D_{0,n}|x^n),~D_{0,n}\in{\cal B}({\cal Y}_{0,n}).\label{equation4b}
\end{align}
\noi Then, (\ref{equation4b}) is a unique measure on $({\cal Y}^{\mathbb{N}},{\cal B}({\cal Y}^{\mathbb{N}}))$ from which $\{q_n(dy_n;y^{n-1},x^n):n\in\mathbb{N}\}$ is obtained.\\
An equivalent definition of a feedforward channel is a family of measures ${\bf Q}(D|{\bf x})$ satisfying the following consistency condition.\\
{\bf C2}: If $F\in{\cal B}({\cal Y}_{0,n}),$ then ${\bf Q}(F|{\bf x})$ is ${\cal B}({\cal X}_{0,n})-$measurable function of ${\bf x}\in{\cal X}^{\mathbb{N}}.$\\
The set of such measures is denoted by ${\cal Q}^{\bf C2}({\cal Y}^{\mathbb{N}};{\cal X}^{\mathbb{N}})$. Then, for any family of measures ${\bf Q}(\cdot|{\bf x})$ on $({\cal Y}^{\mathbb{N}},{\cal B}({\cal Y}^{\mathbb{N}}))$ satisfying {\bf C2} one can construct a collection of conditional distributions $\{q_n(dy_n;y^{n-1},x^n):n\in\mathbb{N}\}$ which are connected with ${\bf Q}(\cdot|{\bf x})$ via relation (\ref{equation4}).

%%%%%%%%%%%%%%%%%%%%%%%%%%%%%%%%%%%%%%%%%%%%%%%%%%%%%%%%%%%%%%%%%%%%%%%%%%%%%%%%%%%%%%%%%%%%

\subsection{Directed Information Functional}\label{directed_information}
\par Next, we  define directed information $I(X^n\rightarrow{Y^n})$ using  ${\bf P}(\cdot|{\bf y})$  and ${\bf Q}(\cdot|{\bf x})$. Given ${\bf P}(\cdot|\cdot)\in{\cal Q}^{\bf C1}({\cal X}^{\mathbb{N}};{\cal Y}^{\mathbb{N}})$ and ${\bf Q}(\cdot|\cdot)\in{\cal Q}^{\bf C2}({\cal Y}^{\mathbb{N}};{\cal X}^{\mathbb{N}})$ define:\\
{\bf P1}: The joint distribution on ${\cal X}^{\mathbb{N}}\times{\cal Y}^{\mathbb{N}}$ defined uniquely by
\begin{align}
({\overleftarrow P}_{0,n}\otimes{\overrightarrow Q}_{0,n})(\times^n_{i=0}A_i{\times}B_i),A_i\in{\cal B}({\cal X}_i),~B_i\in{\cal B}({\cal Y}_i). \nonumber 
\end{align}
{\bf P2}: The marginal distributions on ${\cal X}^{\mathbb{N}}$ defined uniquely for $A_i\in{\cal B}({\cal X}_i)$, $i=0,1,\ldots,n$, by
\begin{align}
\mu_{0,n}(\times^n_{i=0}A_i)%\tri\mathbb{P}\{X_0\in{A}_0, Y_0\in{\cal Y}_0,\ldots, X_n\in{A}_n, Y_n\in{\cal Y}_n\},\nonumber\\
=({\overleftarrow P}_{0,n}\otimes{\overrightarrow Q}_{0,n})(\times^n_{i=0}(A_i\times{\cal Y}_i)).\nonumber%\label{equation19}
\end{align}
{\bf P3}: The marginal distributions on ${\cal Y}^{\mathbb{N}}$ defined uniquely for $B_i\in{\cal B}({\cal Y}_i),~i=0,1,\ldots,n$, by
\begin{align}
\nu_{0,n}(\times^n_{i=0}B_i)=({\overleftarrow P}_{0,n}\otimes{\overrightarrow Q}_{0,n})(\times^n_{i=0}({\cal X}_i\times{B}_i)).\nonumber%\label{equation20}
\end{align}
{\bf P4}: The measure ${\overrightarrow\Pi}_{0,n}:{\cal B}({\cal X}_{0,n})\odot{\cal B}({\cal Y}_{0,n})\mapsto[0,1]$ defined uniquely for $A_i\in{\cal B}({\cal X}_i)$,~$B_i\in{\cal B}({\cal Y}_i),~i=0,1,\ldots,n$, by
\begin{align}
{\overrightarrow\Pi}_{0,n}(\times^n_{i=0}(A_i{\times}B_i))&\tri({\overleftarrow P}_{0,n}\otimes\nu_{0,n})(\times^n_{i=0}(A_i{\times}B_i)).\nonumber%\label{equation21}
\end{align}
\noi By invoking the definition of directed information  and measures {\bf P1}-{\bf P4}, it can be shown by repeated application of chain rule of relative entropy \cite{dupuis-ellis97} that\footnote{Unless stated otherwise, integrals with respect to measures are over the spaces on which these are defined.}
\begin{align}
&I(X^n\rightarrow{Y}^n)=\mathbb{D}({\overleftarrow P}_{0,n} \otimes {\overrightarrow Q}_{0,n}||{\overrightarrow\Pi}_{0,n})\label{equation33}\\
&=\int \log \Big( \frac{{\overrightarrow Q}_{0,n}(d y^n|x^n)}{\nu_{0,n}(dy^n)}\Big)({\overleftarrow P}_{0,n}\otimes {\overrightarrow Q}_{0,n})(dx^n,dy^n)\label{equation203}\\
&\equiv{\mathbb{I}}_{X^n\rightarrow{Y^n}}({\overleftarrow P}_{0,n}, {\overrightarrow Q}_{0,n}).\label{equation7a}
\end{align}
 The notation ${\mathbb{I}}_{X^n\rightarrow{Y^n}}(\cdot,\cdot)$ indicates the functional dependence of $I(X^n\rightarrow{Y^n})$ on $\{{\overleftarrow P}_{0,n}, {\overrightarrow Q}_{0,n}\}$.

%%%%%%%%%%%%%%%%%%%%%%%%%%%%%%%%%%%%%%%%%%%%%%%%%%%%%%%%%%%%%%%%%%%%%%%%%%%%%%%%%%%%%%%%%%%%%

\section{Variational Equalities}\label{variational}

\par In this section we derive two variational equalities associated with $I(X^n\rightarrow{Y^n})$. First, we recall one of the variational equalities of mutual information $I(X^n;Y^n)\equiv\mathbb{I}_{X^n;Y^n}(P_{X^n},P_{Y^n|X^n})$ which can be expressed as maximization of relative entropy functionals as follows \cite{blahut1987}.\\
{\bf Max:} Given a channel ${P}_{Y^n|X^n}(dy^n|x^n)$, a source $P_{X^n}(dx^n)$, and any conditional distribution $\bar{P}_{X^n|Y^n}(dx^n|y^n)$ then
\begin{align}
&\mathbb{I}_{X^n;Y^n}(P_{X^n},P_{Y^n|X^n})=\sup_{\bar{P}_{X^n|Y^n}}\int\log\bigg(\frac{\bar{P}_{X^n|Y^n}(dx^n|y^n)}{P_{X^n}(dx^n)}\bigg)\nonumber\\
&\qquad\qquad\times{P}_{Y^n|X^n}(dy^n|x^n)\otimes P_{X^n}(dx^n)\label{equation2c}
\end{align}
and the supremum is achieved at $\bar{P}_{X^n|Y^n}(dx^n|y^n)=\frac{P_{Y^n|X^n}(dy^n|x^n)\otimes{P}_{X^n}(dx^n)}{\int_{{\cal X}_{0,n}}P_{Y^n|X^n}(dy^n|x^n)\otimes{P}_{X^n}(dx^n)}$.\\
\noi Let ${\bf P}(\cdot|\cdot)\in{\cal Q}^{\bf C1}({\cal X}^{\mathbb{N}};{\cal Y}^{\mathbb{N}})$ and ${\bf Q}(\cdot|\cdot)\in{\cal Q}^{\bf C2}({\cal Y}^{\mathbb{N}};{\cal X}^{\mathbb{N}})$, and let $P_{0,n}(dx^n,dy^n)={\overleftarrow P}_{0,n}(dx^n|y^{n-1})\otimes{\overrightarrow Q}_{0,n}(dy^n|x^n)$.\\
Next we derive the analogous version for directed information. Let ${\bf S}(\cdot|{\bf x})$ be any measure on $({\cal Y}^{\mathbb{N}},{\cal B}({\cal Y}^{\mathbb{N}}))$ satisfying the consistency condition\\
{\bf C3:} If $F\in{\cal B}({\cal Y}_{0,n}),$ then ${\bf S}(F|{\bf x})$ is a ${\cal B}({\cal X}_{0,n-1})-$measurable.\\
Denote this family of measures by ${\bf S}(\cdot|{\bf x})\in{\cal Q}^{\bf C3}({\cal Y}^{\mathbb{N}};{\cal X}^{\mathbb{N}})$. By Section~\ref{causal_channels}, for any family of measures ${\bf S}(\cdot|{\bf x})$  satisfying consistency condition {\bf C3}, there exists a collection  $\{s_n(\cdot;\cdot,\cdot)\in{\cal Q}({\cal Y}_n;{\cal Y}_{0,n-1}\times{\cal X}_{0,n-1}):n\in\mathbb{N}\}$ connected to ${\bf S}(\cdot|{\bf x})$ by
\begin{align}
&{\bf S}(D|{\bf x})=\int_{D_0}s_0(dy_0)\ldots\int_{D_n}s_n(dy_n;y^{n-1},x^{n-1})\nonumber\\
&\equiv\overleftarrow{S}_{0,n}(D_{0,n}|x^{n-1}),~D_{0,n}\tri\times_{i=0}^n{D}_i\in{\cal B}({\cal Y}_{0,n}).\label{equation110}
\end{align}
Unlike $\overrightarrow{Q}_{0,n}(\cdot|x^n)$ which is conditioned on $x^n\in{\cal X}_{0,n}$, the measure $\overleftarrow{S}_{0,n}(\cdot|x^{n-1})$ is conditioned on $x^{n-1}\in{\cal X}_{0,n-1}$.\\
Let ${\bf R}(\cdot|{\bf y})$ be any family of measures on $({\cal X}^{\mathbb{N}},{\cal B}({\cal X}^{\mathbb{N}}))$ satisfying the consistency condition\\
{\bf C4:} If $E\in{\cal B}({\cal X}_{0,n}),$ then ${\bf R}(E|{\bf y})$ is a ${\cal B}({\cal Y}_{0,n})-$measurable.\\
Denote this family of measures by ${\bf R}(\cdot|{\bf y})\in{\cal Q}^{\bf C4}({\cal X}^{\mathbb{N}};{\cal Y}^{\mathbb{N}})$. Similarly as before, for any family of measures ${\bf R}(\cdot|{\bf y})$  satisfying consistency condition {\bf C4}, there exists  $\{r_n(\cdot;\cdot,\cdot)\in{\cal Q}({\cal X}_n;{\cal X}_{0,n-1}\times{\cal Y}_{0,n}):n\in\mathbb{N}\}$ connected to ${\bf R}(\cdot|{\bf y})$ by
\begin{align}
&{\bf R}(G|{\bf y})=\int_{G_0}r_0(dx_0;y_0)\ldots\int_{G_n}r_n(dx_n;x^{n-1},y^{n})\nonumber\\
&\equiv{\overrightarrow{R}}_{0,n}(G_{0,n}|y^n),~G_{0,n}\tri\times_{i=0}^n{G}_i\in{\cal B}({\cal X}_{0,n}).\label{equation111}
\end{align}
Unlike $\overleftarrow{P}_{0,n}(\cdot|y^{n-1})$ which is conditioned on $y^{n-1}\in{\cal Y}_{0,n}$, the measure $\overrightarrow{R}_{0,n}(\cdot|y^{n})$ is conditioned on $y^{n}\in{\cal Y}_{0,n}$.\\
Define another joint distribution on $\big{(}{\cal X}^{\mathbb{N}}\times{\cal Y}^{\mathbb{N}},\odot_{n\in\mathbb{N}}{\cal B}({\cal X}_n)\odot{\cal B}({\cal Y}_n)\big{)}$ by $({\overleftarrow S}_{0,n}\otimes{\overrightarrow R}_{0,n})(dx^n,dy^n)$.\\
\noi The next theorem gives the two variational equalities.
\begin{theorem}(Variational Equalities)\label{variational_equalities}{\ \\}
{\bf Part A.} For any arbitrary measure $\bar{\nu}_{0,n}\in{\cal M}_1({\cal Y}_{0,n})$
\begin{align}
&{\mathbb{I}}_{X^n\rightarrow{Y^n}}({\overleftarrow P}_{0,n}, {\overrightarrow Q}_{0,n})\tri\mathbb{D}(P_{0,n}||{\overrightarrow\Pi}_{0,n})\nonumber\\
&=\inf_{\bar{\nu}_{0,n}\in{\cal M}_1({\cal Y}_{0,n})}\mathbb{D}({\overleftarrow P}_{0,n} \otimes {\overrightarrow Q}_{0,n}||{\overleftarrow P}_{0,n} \otimes\bar{\nu}_{0,n} )\label{equation11}\\
&=\inf_{\bar{\nu}_{0,n}\in{\cal M}_1({\cal Y}_{0,n})}\int\log \Big( \frac{{\overrightarrow Q}_{0,n}(dy^n|x^n)}{\bar{\nu}_{0,n}(dy^n)}\Big)\nonumber\\
&\qquad\qquad\times({\overleftarrow P}_{0,n}\otimes {\overrightarrow Q}_{0,n})(dx^n,dy^n)\label{equation10}
\end{align}
and the infimum in (\ref{equation11}) is achieved at $\bar{\nu}_{0,n}^*(dy^n)=\int_{{\cal X}_{0,n}}({\overleftarrow P}_{0,n}\otimes {\overrightarrow Q}_{0,n})(dx^n,{dy^n})\equiv {\nu}_{0,n}(dy^n)$.\\
{\bf Part B.} For any ${\bf S}(\cdot|{\bf x})\in{\cal Q}^{\bf C3}({\cal Y}^{\mathbb{N}};{\cal X}^{\mathbb{N}})$ and ${\bf R}(\cdot|{\bf y})\in{\cal Q}^{\bf C4}({\cal X}^{\mathbb{N}};{\cal Y}^{\mathbb{N}})$ then
\begin{align}
&{\mathbb{I}}_{X^n\rightarrow{Y^n}}({\overleftarrow P}_{0,n},{\overrightarrow Q}_{0,n})=\mathbb{D}(P_{0,n}||{\overrightarrow \Pi}_{0,n})\nonumber\\
&=\sup_{   \overleftarrow{S}_{0,n} \otimes \overrightarrow{R}_{0,n} }\int\log \Big( \frac{d({\overleftarrow S}_{0,n}\otimes {\overrightarrow R}_{0,n})}{d ( {\overrightarrow \Pi}_{0,n} ) }\Big){d}({\overleftarrow P}_{0,n}\otimes {\overrightarrow Q}_{0,n})\label{equation34}
\end{align}
and the supremum in (\ref{equation34}) is achieved when the RND satisfies
\begin{align}
\Lambda_{0,n}(x^n,y^n)\tri\frac{d({\overleftarrow P}_{0,n}\otimes {\overrightarrow Q}_{0,n})}{d({\overleftarrow S}_{0,n}\otimes {\overrightarrow R}_{0,n})}=1-a.s.,~n\in\mathbb{N}.\label{equation105}
\end{align}
Equivalently, for $i=0,1,\ldots,n$,
\begin{align}
\lambda_i(x^i,y^i)\tri\frac{p_i(dx_i;x^{i-1},y^{i-1})\otimes{q}_i(dy_i;y^{i-1},x^{i})}{s_i(dy_i;y^{i-1},x^{i-1})\otimes{r}_i(dx_i;x^{i-1},y^i)}=1-a.s.\label{equation102}
\end{align}
\end{theorem}
\begin{proof}
The derivation is shown in detail in \cite{charalambous-stavrou2013a}.
\end{proof}

\noi{\bf Discussion.} Clearly, (\ref{equation2c}) is also equivalent to (i.e., $P_{Y^n}$ is fixed)
\begin{align}
&\sup_{\bar{P}_{X^n|Y^n}\otimes{P}_{Y^n}}\int\log\bigg(\frac{\bar{P}_{X^n|Y^n}(dx^n|y^n)\otimes{P}_{Y^n}(dy^n)}{P_{X^n}(dx^n)\times{P}_{Y^n}(dy^n)}\bigg)\nonumber\\
&\qquad\qquad\times{P}_{Y^n|X^n}(dy^n|x^n)\otimes{P}_{X^n}(dx^n)\label{equation112}
\end{align}
since the RND in (\ref{equation112}) is another version of the one in (\ref{equation2c}). Thus, (\ref{equation34}) is the analogue of (\ref{equation112}), in which the directed information function is utilized together with the decomposition $\overleftarrow{S}_{0,n}\otimes\overrightarrow{R}_{0,n}$ of the joint distribution. Suppose $q_i(\cdot;y^{i-1},x^i)\ll{s}_i(\cdot;y^{i-1},x^{i-1}), \forall i$  and $\overleftarrow{S}_{0,n}$ is  fixed,  generated by ${\bf P}(\cdot|\cdot)\in{\cal Q}^{\bf C1}({\cal X}^{\mathbb{N}};{\cal Y}^{\mathbb{N}})$ and ${\bf Q}(\cdot|\cdot)\in{\cal Q}^{\bf C2}({\cal Y}^{\mathbb{N}};{\cal X}^{\mathbb{N}})$.  Then from (\ref{equation102}):
\begin{align}
&r_i(dx_i;x^{i-1},y^i)=    
\frac{q_i(dy_i;y^{i-1},x^i)\otimes{p}_i(dx_i;x^{i-1},y^{i-1})}{\int_{{\cal X}_i}q_i(dy_i;y^{i-1},x^i)\otimes{p}_i(dx_i;x^{i-1},y^{i-1})} \nonumber \\
&\overrightarrow{R}_{0,n}(\cdot|y^n)= \otimes_{i=0}^n\frac{q_i(dy_i;y^{i-1},x^i)\otimes{p}_i(dx_i;x^{i-1},y^{i-1})}{\int_{{\cal X}_i}q_i(dy_i;y^{i-1},x^i)\otimes{p}_i(dx_i;x^{i-1},y^{i-1})}.\nonumber
\end{align}
The previous expression is the analogue of $\bar{P}_{X^n|Y^n}$ in (\ref{equation2c}). 

%%%%%%%%%%%%%%%%%%%%%%%%%%%%%%%%%%%%%%%%%%%%%%%%%%%%%%%%%%%%%%%%%%%%%%%%%%%%%%%%%%%%%%%%%%%

\section{Applications to Nonanticipative RDF}\label{applications}
%%%%%%%%%%%%%%%%%%%%%%%%%%%%%%%%%%%%%%%%%%%%%%%%%%%%%%%%%%%%%%%%%%%%%%%%
Our interest is now focused on nonanticipative RDF, which is motivated by source-channel matching via symbol-by-symbol transmission, for sources with memory. Unlike classical RDF, the solution of nonanticipative RDF is causal and hence, it can be realized by  an encoder-channel-decoder which process information causally \cite{kourtellaris-charalambous-stavrou2013a,charalambous-stavrou-ahmed2012a}. Moreover, as it is shortly shown, nonanticipative RDF is relatively easy to  compute when compared to the classical RDF (which is only computed for a small class of sources, i.e., memoryless and Gaussian).
\par First, we recall Gorbunov-Pinsker's definition of nonanticipatory $\epsilon$-entropy \cite{gorbunov-pinsker} since we will establish its equivalence to the nonanticipative RDF.
Introduce the measurable distortion function by $d_{0,n}(x^n,y^n):{\cal X}_{0,n}\times{\cal Y}_{0,n}\mapsto[0,\infty)$, $d_{0,n}(x^n,y^n)\tri\sum_{i=0}^n\rho_{0,i}(x^i,y^i)$, and let $d_{0,n}(x^n,y^n)\tri\sum_{i=0}^n\rho(x_i,y_i)$ for single letter.
Introduce the fidelity set by
\begin{align*}
&{\cal Q}_{0,n}(D)\tri\Big\{P_{Y^n|X^n}(dy^n|x^n):\\
&\frac{1}{n+1}\int{d}_{0,n}(x^n,y^n)P_{Y^n|X^n}(dy^n|x^n)\otimes{P}_{X^n}(dx^n)\leq{D}\Big\}.
\end{align*}
\noi Gorbunov and Pinsker restricted the set ${\cal Q}_{0,n}(D)$ to those reproduction distributions which satisfy the Markov chain (MC) $X_{n+1}^\infty\leftrightarrow{X^n}\leftrightarrow{Y^n}\Leftrightarrow{P}_{Y^n|X^{\infty}}(dy^n|x^{\infty})={P}_{Y^n|X^n}(dy^n|x^n)-a.s., \forall n\geq 0$. Then they introduced the nonanticipatory $\epsilon$-entropy defined by
\begin{align}
R_{0,n}^{\varepsilon}(D)\tri\inf_{\substack{{\cal Q}_{0,n}(D):~X^n_{i+1}\leftrightarrow{X^i}\leftrightarrow{Y^i}\\i=0,1,\ldots,n-1}}I(X^n;Y^n).\label{equation15}
\end{align} 
Thus, the difference between the classical RDF and nonanticipatory $\epsilon$-entropy (\ref{equation15}) is the presence of the MC which implies that for each $i$, $Y_i$ is a function of the past and present source symbols $\{X_0,X_1,\ldots,X_i\}$, and independent of the future source symbols $\{X_{i+1},\ldots,X^{n}\}$.
It can be shown that the MC $X_{i+1}^n\leftrightarrow{X^i}\leftrightarrow{Y^i},~i=0,1,\ldots,n-1$, is equivalent to ${P}_{Y^i|X^i}(dy^i|x^i)=\overrightarrow{P}_{Y^i|X^i}(dy^i|x^i)-a.s.$, $i=0,1,\ldots,n-1$. Utilizing this MC, then 
\begin{align}
&I(X^n;{Y^n})=\int\log\Big(\frac{{\overrightarrow P}_{Y^n|X^n}(dy^n|x^n)}{{P}_{Y^n}(dy^n)}\Big)\nonumber\\
&\times{\overrightarrow P}_{Y^n|X^n}(dy^n|x^n)\otimes{P}_{X^n}(dx^n)\equiv{\mathbb I}_{X^n\rightarrow{Y^n}}(P_{X^n},{\overrightarrow P}_{Y^n|X^n}) \nonumber 
\end{align}
where the notation ${\mathbb I}_{X^n\rightarrow{Y^n}}(P_{X^n},{\overrightarrow P}_{Y^n|X^n})$ is used to point out the functional dependence on $\{P_{X^n},{\overrightarrow P}_{Y^n|X^n}\}$. Moreover, utilizing the previous expression it is easy to show that nonanticipatory $\epsilon$-entropy (\ref{equation15}) is equivalent to the following definition of nonanticipative RDF.

\begin{definition}(Nonanticipative RDF) 
Let $\overrightarrow{\cal Q}_{0,n}(D)$ (assuming is non-empty) denotes the fidelity set
\begin{align*}
&\overrightarrow{\cal Q}_{0,n}(D)\tri\Big\{{\overrightarrow P}_{Y^n|X^n}(y^n|x^{n}):\ell_{d_{0,n}}(\overrightarrow{P}_{Y^n|X^n})\tri \nonumber \\
&\frac{1}{n+1}\int{d}_{0,n}(x^n,y^{n}){\overrightarrow P}_{Y^n|X^n}(dy^n|x^n)\otimes{P}_{X^n}(dx^n)\leq{D}\Big\}.
\end{align*}
The nonanticipative information RDF is defined by
\begin{align}
{R}^{na}_{0,n}(D)=\inf_{\overrightarrow{P}_{Y^n|X^n}\in\overrightarrow{\cal Q}_{0,n}(D)}{\mathbb I}_{X^n\rightarrow{Y^n}}(P_{X^n},{\overrightarrow P}_{Y^n|X^n}).\label{ex12} 
\end{align}
\end{definition}

Next, we introduce some assumptions and we establish existence of the infimum in (\ref{ex12}) and hence also of (\ref{equation15}).
\begin{assumption}(Main asumptions)\label{conditions-existence}{\ \\}
{\bf(A1)} ${\cal Y}_{0,n}$ is a compact Polish space, ${\cal X}_{0,n}$ is a Polish space;\\
{\bf(A2)} for all $h(\cdot){\in}BC({\cal Y}_{0,n})$,  $(x^{n},y^{n-1})\in{\cal X}_{0,n}\times{\cal Y}_{0,n-1}\mapsto\int_{{\cal Y}_n}h(y)P_{Y|Y^{n-1},X^n}(dy|y^{n-1},x^n)\in\mathbb{R}$ 
is continuous jointly in  $(x^{n},y^{n-1})\in{\cal X}_{0,n}\times{\cal Y}_{0,n-1}$;\\
{\bf(A3)} $d_{0,n}(x^n,\cdot)$ is continuous on ${\cal Y}_{0,n}$;\\
{\bf(A4)} There exist  $(x^n,y^{n})\in{\cal X}_{0,n}\times{\cal Y}_{0,n}$ such that  $d_{0,n}(x^n,y^{n})<D$.
\end{assumption}
\noi Note that since ${\cal Y}_{0,n}$ is assumed to be a compact Polish space, then by \cite{dupuis-ellis97}, probability measures on ${\cal Y}_{0,n}$ are weakly compact. Moreover, the following result can be obtained, which we will use to show existence of the infimum in (\ref{ex12}).
\begin{lemma}\label{compactness2}\cite{stavrou-charalambous2013}
Suppose Assumption~\ref{conditions-existence}, {\bf(A1)}, {\bf(A2)} hold. Then\\
{\bf(1)} The set ${\cal Q}^{\bf C2}({\cal Y}_{0,n};{\cal X}_{0,n})$ is weakly compact.\\
{\bf(2)} $\mathbb{I}_{X^n\rightarrow{Y^n}}(P_{X^n},\overrightarrow{P}_{Y^n|X^n})$ is lower semicontinuous on ${\cal Q}^{\bf C2}({\cal Y}_{0,n};{\cal X}_{0,n})$ for a fixed ${\cal M}_1({\cal X}_{0,n})$.\\
{\bf(3)} Under the additional Assumption~\ref{conditions-existence},~{\bf(A3)}, {\bf(A4)}  the set ${\overrightarrow{\cal Q}}_{0,n}(D)$ is a closed subset of ${\cal Q}^{\bf C2}({\cal Y}_{0,n};{\cal X}_{0,n})$.
\end{lemma}
\noi The next theorem establishes existence of the minimizing reproduction distribution for (\ref{ex12}).
\begin{theorem}\cite{stavrou-charalambous2013}(Existence)\label{existence_rd}
Suppose Assumption~\ref{conditions-existence} hold. Then the infimum in (\ref{ex12}) is achieved and ${R}^{na}_{0,n}(D)$ is finite. 
\end{theorem}
\begin{remark}(Summary) 
Utilizing Theorem~\ref{existence_rd} and \cite[Theorems 3, 4]{gorbunov-pinsker}, for a stationary source and single letter distortion, $\lim_{n\rightarrow\infty}{R}^{na}_{0,n}(D)$ exists, it is finite, and the optimal reproduction distribution is realizable by stationary source-reproduction pairs $\{(X_i,Y_i):~i=0,1,\ldots,n\}$. Hence, the $(n+1)$-fold convolution conditional distribution $\overrightarrow{P}_{Y^n|X^n}(dy^n|x^n)=\otimes^n_{i=0}P_{Y_i|Y^{i-1},X^i}$ $(dy_i|y^{i-1},x^i)-a.s.$, is a convolution of stationary conditional distributions.
\end{remark}
\subsection{Optimal Stationary Reproduction Distribution of Nonanticipative RDF}

\par Next, we give the solution of ${R}^{na}_{0,n}(D)$ assuming the reproduction is stationary. 
By utilizing the Lagrange duality theorem we obtain the unconstrained problem.
\begin{align}
{R}^{na}_{0,n}(D) &= \sup_{s\leq{0}}\inf_{\substack{{\overrightarrow{P}_{Y^n|X^n}}\\
\in{Q}^{\bf C2}({\cal Y}_{0,n};{\cal X}_{0,n})}} \Big\{{\mathbb I}_{X^n\rightarrow{Y^n}}(P_{X^n},\overrightarrow{P}_{Y^n|X^n})\nonumber\\
&-s(\ell_{{d}_{0,n}}(\overrightarrow{P}_{Y^n|X^n})-D)\Big\}. \label{ex13}
\end{align}
Note that ${\overrightarrow{P}_{Y^n|X^n}} \in {\cal Q}^{\bf C2}({\cal Y}_{0,n};{\cal X}_{0,n})$ are probability measures  therefore, one should introduce another set of Lagrange multipliers to obtain an unconstrained problem free of such a constraint. For the rest of the paper, we consider  $d_{0,n}(x^n,y^n)\tri\sum_{i=0}^n\rho(T^i{x^n},T^i{y^n})$, where $T^i{x^n}$ is the shift operator on $x^n$ (similarly for $T^i{y^n}$). Then by computing the Gateaux differential of (\ref{ex13}) (for the stationary case) we obtain the following (see  \cite{stavrou-charalambous2013}).\\ 
{\bf{(1)}} The infimum in (\ref{ex13}) is attained at  $\overrightarrow{P}^*_{Y^n|X^n} \in\overrightarrow{\cal Q}_{0,n}(D)$: 
\begin{align}
&\overrightarrow{P}^*_{Y^n|X^n}(dy^n|x^n)=\otimes_{i=0}^n{P}^*_{Y_i|Y^{i-1},X^i}(dy_i|y^{i-1},x^i)\nonumber\\
&=\otimes_{i=0}^n\frac{e^{s \rho(T^i{x^n},T^i{y^n})}P^*_{Y_i|Y^{i-1}}(dy_i|y^{i-1})}{\int_{{\cal Y}_i} e^{s \rho(T^i{x^n},T^i{y^n})} P^*_{Y_i|Y^{i-1}}(dy_i|y^{i-1})}, \: s\leq 0.\label{ex14}
\end{align}
{\bf{(2)}} The nonanticipative RDF is given by
\begin{align}
&{R}^{na}_{0,n}(D)=sD(n+1) -\sum_{i=0}^n\int\log \Big( \int_{{\cal Y}_i} e^{s\rho(T^i{x^n},T^i{y^n})} \nonumber\\
&P^*_{Y_i|Y^{i-1}}(dy_i|y^{i-1})\Big){\overrightarrow{P}^*_{Y^{i-1}|X^{i-1}}(dy^{i-1}|x^{i-1})\otimes{P}_{X^i}(dx^i)}.\nonumber 
\end{align}
If ${R}^{na}_{0,n}(D) > 0$ then $ s < 0$  and
\begin{align*}
\frac{1}{n+1}\sum_{i=0}^n\int\rho(T^i{x^n},T^i{y^n})(P_{X^i}\otimes\overrightarrow{P}^*_{Y^{i}|X^{i}})(dx^i,dy^i)=D.
\end{align*}
By (\ref{ex14})  the optimal reproduction distribution  is nonanticipative (causal) with respect to the source, and by stationarity all elements in the product are identical. Moreover, if the distortion $\rho(T^ix^n,T^iy^n)=\rho(x_i,T^iy^n)$,~$\forall{i}$, then $P^*_{Y_i|Y^{i-1},X^i}(dy_i|y^{i-1},x^i)=P^*_{Y_i|Y^{i-1},X_i}(dy_i|y^{i-1},x_i)-a.s.,~\forall{i}$, i.e.,  it depends only on the most recent source symbol.
%%%%%%%%%%%%%%%%%%%%%%%%%%%%%%%%%%%%%%%%%%%%%%%%%%%%%%%%%%%%%%%%%
\subsection{BAA for Stationary Nonanticipative RDF}
\par Now, we invoke variational equality (\ref{equation10}), and (\ref{ex14}), to compute BAA for the nonanticipative RDF (stationary case).
\begin{theorem}(Double Minimization)
\label{double_minimization}{\ \\}
{\bf{(a)}} ${R}^{na}_{0,n}(D)$, can be expressed as a double minimization:
\begin{align}
&{R}^{na}_{0,n}(D)=sD(n+1)+\min_{\bar{P}_{Y^n}\in{\cal M}_1({\cal Y}_{0,n})}\min_{\overrightarrow{P}_{Y^n|X^n}\in{\cal Q}^{\bf C2}({\cal Y}_{0,n};{\cal X}_{0,n})}\nonumber\\
&\Big{\{}\int\log\Big(\frac{\overrightarrow{P}_{Y^n|X^n}(dy^n|x^n)}{\bar{P}_{Y^n}(dy^n)}\Big)(P_{X^n}\otimes\overrightarrow{P}_{Y^n|X^n})(dx^n,dy^n)\nonumber\\
&-s\int{d}_{0,n}(x^n,y^n)(P_{X^n}\otimes\overrightarrow{P}_{Y^n|X^n})(dx^n,dy^n)\Big{\}}.\label{eq.13}
\end{align}
{\bf{(b)}} For fixed $\overrightarrow{P}_{Y^n|X^n}$, the minimization over $\bar{P}_{Y^n}\in{\cal M}_1({\cal Y}_{0,n})$ is
\begin{align*}
\bar{P}_{Y^n}^*(dy^n)=\int(P_{X^n}\otimes\overrightarrow{P}_{Y^n|X^n})(dx^n,dy^n).
\end{align*}
{\bf{(c)}} For fixed $\bar{P}_{Y^n}$, the minimization over $\overrightarrow{P}_{Y^n|X^n}\in{\cal Q}^{\bf C2}({\cal Y}_{0,n};{\cal X}_{0,n})$ is given by (\ref{ex14}).
\end{theorem} 
\begin{proof} 
A consequence of previous section.
\end{proof}

We now have the following algorithm. 
\begin{theorem}(Convergence of BAA)
Let $\bar{P}^0_{Y^n}$ be any probability measure which is positive. Let  $\bar{P}_{Y^n}^{r+1}$ be given in terms of $\bar{P}_{Y^n}^r$ by
\begin{align}
\bar{P}_{Y^n}^{r+1}=\bar{P}_{Y^n}^r\int\bigg(\otimes_{i=0}^n\frac{A_{i}}{\int_{{\cal Y}_i}A_{i}\bar{P}^{r}_{Y_i|Y^{i-1}}(dy_i|y^{i-1})}\bigg)P_{X^n}(dx^n)\nonumber
\end{align}
where $A_{i}=e^{s\rho(T^i{x^n},T^i{y^n})}$. Then
\begin{align}
D(\overrightarrow{P}_{Y^n|X^n}(\bar{P}^r_{Y^n}))&\longrightarrow{D}_s, ~\mbox{as}~r\rightarrow\infty\nonumber\\
\mathbb{I}_{X^n\rightarrow{Y^n}}(P_{X^n},\overrightarrow{P}_{Y^n|X^n}(\bar{P}_{Y^n}^r))&\longrightarrow{R}^{na}_{0,n}(D_s),~\mbox{as}~r\rightarrow\infty\nonumber
\end{align}
where $D_s=\int{d_{0,n}(x^n,y^n)}(P_{X^n}\otimes\overrightarrow{P}_{Y^n|X^n})(dx^n,dy^n)$ and $(D_s,{R}^{na}_{0,n}(D_s))$ is a point on the curve ${R}^{na}_{0,n}(D)$ parametrized by $s$.
\end{theorem}
\begin{proof}
The derivation utilizes Theorem~\ref{double_minimization} and \cite{blahut1987}.
\end{proof}
%%%%%%%%%%%%%%%%%%%%%%%%%%%%%%%%%%%%%%%%%%%%%%%%%%%%%%%%%%%%%%%%%%%%%%%%%%%%%%%%%%%%%%%%%%%%%%

\section{Conclusion}
In this paper we derive two variational equalities for directed information. Then we show existence of the reproduction distribution which achieves the infimum of the nonanticipative RDF, and we use the variational equality to find a BAA for stationary nonanticipative RDF. Recently, we have applied the nonanticipative RDF in source-channel matching via symbol-by-symbol transmission. Specifically, we have computed $R^{na}(D)$ explicitly for sources with memory without anticipation \cite{kourtellaris-charalambous-stavrou2013a}, and we have used $R^{na}(D)$ in filtering applications \cite{charalambous-stavrou-ahmed2012a,stavrou-charalambous2013}.

\bibliographystyle{IEEEtran}
\bibliography{photis_references_variational_equalities}

\begin{thebibliography}{10}
\providecommand{\url}[1]{#1}
\csname url@samestyle\endcsname
\providecommand{\newblock}{\relax}
\providecommand{\bibinfo}[2]{#2}
\providecommand{\BIBentrySTDinterwordspacing}{\spaceskip=0pt\relax}
\providecommand{\BIBentryALTinterwordstretchfactor}{4}
\providecommand{\BIBentryALTinterwordspacing}{\spaceskip=\fontdimen2\font plus
\BIBentryALTinterwordstretchfactor\fontdimen3\font minus
  \fontdimen4\font\relax}
\providecommand{\BIBforeignlanguage}[2]{{%
\expandafter\ifx\csname l@#1\endcsname\relax
\typeout{** WARNING: IEEEtran.bst: No hyphenation pattern has been}%
\typeout{** loaded for the language `#1'. Using the pattern for}%
\typeout{** the default language instead.}%
\else
\language=\csname l@#1\endcsname
\fi
#2}}
\providecommand{\BIBdecl}{\relax}
\BIBdecl

\bibitem{charalambous-stavrou-ahmed2012a}
\BIBentryALTinterwordspacing
C.~D. Charalambous, P.~A. Stavrou, and N.~U. Ahmed, ``Nonanticipative rate
  distortion function and relations to filtering theory,'' \emph{submitted to
  IEEE Transactions on Automatic Control}, 2013. [Online]. Available:
  \url{http://arxiv.org/abs/1210.1266v2.}
\BIBentrySTDinterwordspacing

\bibitem{chen-berger2005}
J.~Chen and T.~Berger, ``The capacity of finite-state {M}arkov channels with
  feedback,'' \emph{IEEE Transactions on Information Theory}, vol.~51, no.~3,
  pp. 780--798, {M}ar. 2005.

\bibitem{tatikonda-mitter2009}
S.~Tatikonda and S.~Mitter, ``The capacity of channels with feedback,''
  \emph{IEEE Transactions on Information Theory}, vol.~55, no.~1, pp. 323--349,
  {J}an. 2009.

\bibitem{tatikonda2000}
S.~C. Tatikonda, ``Control over communication constraints,'' Ph.D.
  dissertation, Mass. Inst. of Tech.~(M.I.T.), Cambridge,~MA, 2000.

\bibitem{ma-ishwar2011}
N.~Ma and P.~Ishwar, ``On delayed sequential coding of correlated sources,''
  \emph{IEEE Transactions on Information Theory}, vol.~57, no.~6, pp.
  3763--3782, 2011.

\bibitem{kramer1998}
G.~Kramer, ``Directed information for channels with feedback,'' Ph.D.
  dissertation, Swiss Federal Institute of Technology (ETH), December 1998.

\bibitem{charalambous-stavrou2012}
C.~D. Charalambous and P.~A. Stavrou, ``Directed information on abstract
  spaces:~properties and extremum problems,'' in \emph{IEEE International
  Symposium on Information Theory (ISIT)}, {J}uly 1-6 2012, pp. 518--522.

\bibitem{naiss-permuter2013b}
I.~Naiss and H.~H. Permuter, ``Computable bounds for rate distortion with feed
  forward for stationary and ergodic sources,'' \emph{IEEE Transactions on
  Information Theory}, vol.~59, no.~2, pp. 760--781, 2013.

\bibitem{naiss-permuter2013a}
------, ``Extension of the blahut-arimoto algorithm for maximizing directed
  information,'' \emph{IEEE Transactions on Information Theory}, vol.~59,
  no.~1, pp. 204--222, 2013.

\bibitem{gastpar2003}
M.~Gastpar, B.~Rimoldi, and M.~Vetterli, ``To code, or not to code: {L}ossy
  source-channel communication revisited,'' \emph{IEEE Transactions on
  Information Theory}, vol.~49, no.~5, pp. 1147--1158, {M}ay 2003.

\bibitem{kourtellaris-charalambous-stavrou2013a}
\BIBentryALTinterwordspacing
C.~K. Kourtellaris, C.~D. Charalambous, and P.~A. Stavrou, ``Nonanticipative
  rate distortion function for general source-channel matching,''
  \emph{submitted to IEEE Information Theory Workshop (ITW)}, 2013. [Online].
  Available: \url{http://arxiv.org/abs/1304.6528.}
\BIBentrySTDinterwordspacing

\bibitem{stavrou-charalambous2013}
\BIBentryALTinterwordspacing
P.~A. Stavrou and C.~D. Charalambous, ``Nonanticipative rate distortion
  function and filtering theory: {A} weak convergence approach,''
  \emph{submitted to Systems and Control Letters}, 2013. [Online]. Available:
  \url{http://arxiv.org/abs/1212.6643v1.}
\BIBentrySTDinterwordspacing

\bibitem{dupuis-ellis97}
P.~Dupuis and R.~S. Ellis, \emph{A Weak Convergence Approach to the Theory of
  Large Deviations}.\hskip 1em plus 0.5em minus 0.4em\relax John Wiley \& Sons,
  Inc., New York, 1997.

\bibitem{blahut1987}
R.~E. Blahut, \emph{{Principles and Practice of Information Theory}}, ser. in
  Electrical and Computer Engineering.\hskip 1em plus 0.5em minus 0.4em\relax
  Reading, MA: Addison-Wesley Publishing Company, 1987.

\bibitem{charalambous-stavrou2013a}
\BIBentryALTinterwordspacing
C.~D. Charalambous and P.~A. Stavrou, ``Directed information on abstract
  spaces: Properties and variational equalities,'' \emph{submitted to IEEE
  Transactions on Information Theory}, 2013. [Online]. Available:
  \url{http://arxiv.org/abs/1302.3971.}
\BIBentrySTDinterwordspacing

\bibitem{gorbunov-pinsker}
A.~K. Gorbunov and M.~S. Pinsker, ``Nonanticipatory and prognostic epsilon
  entropies and message generation rates,'' \emph{Problems of Information
  Transmission}, vol.~9, no.~3, pp. 184--191, {J}uly-{S}ept. 1973.

\end{thebibliography}

\end{document}